\documentclass[12pt, letter]{article}

\usepackage{algorithm}				
\usepackage[noend]{algpseudocode}	
\usepackage{amsmath}				
\usepackage{amssymb}				
\usepackage{amsthm}				
\usepackage{caption}
\usepackage{subcaption}
\usepackage{color}
\usepackage[margin=1in]{geometry}
\usepackage[hidelinks]{hyperref}
\usepackage{titling}

\usepackage{complexity}
\usepackage{multirow}
\usepackage{stmaryrd}

\usepackage{tikz}
\usetikzlibrary{matrix}

\newcommand*\circled[1]{\tikz[baseline=(char.base)]{\node[shape=circle,draw,inner sep=0.5pt] (char) {#1};}}

\usepackage{pifont}				
\newcommand{\cmark}{\ding{51}}
\newcommand{\xmark}{\ding{55}}

\theoremstyle{plain}				
\newtheorem{theorem}{Theorem}
\newtheorem{lemma}[theorem]{Lemma}
\newtheorem{proposition}[theorem]{Proposition}
\newtheorem{corollary}[theorem]{Corollary}

\theoremstyle{definition}			
\newtheorem{definition}[theorem]{Definition}
\newtheorem{example}[theorem]{Example}

\theoremstyle{remark}			
\newtheorem*{remark}{Remark}

\newcommand{\TDFA}{\textsf{2DFA-4W}}

\newcommand{\TDFATW}{\textsf{2DFA-3W}}

\newcommand{\TDFATWOW}{\textsf{2DFA-2W}}

\newcommand{\TNFA}{\textsf{2NFA-4W}}

\newcommand{\TNFATW}{\textsf{2NFA-3W}}

\newcommand{\TNFATWOW}{\textsf{2NFA-2W}}

\thanksmarkseries{alph}

\title{Decision Problems for Restricted Variants of Two-Dimensional Automata}
\author{Taylor J. Smith \thanks{School of Computing, Queen's University, Kingston, Ontario, Canada. Email: \texttt{\{tsmith,ksalomaa\}@cs.queensu.ca}.} \and Kai Salomaa \thanksmark{1}}
\date{\today}

\begin{document}


\maketitle

\begin{abstract}
A two-dimensional finite automaton has a read-only input head that
	moves in four directions on a finite array of cells labelled
	by symbols of the input alphabet.
	A three-way two-dimensional automaton is prohibited from making 
	upward moves, while a two-way two-dimensional automaton can only move downward
	and rightward.

We show that the language emptiness problem for unary three-way 
nondeterministic two-dimensional automata is \NP-complete, 
and is in \P\ for general-alphabet two-way nondeterministic 
two-dimensional automata. We show that the language equivalence problem 
for two-way deterministic two-dimensional automata is decidable, while both the equivalence and universality problems for two-way nondeterministic two-dimensional automata are undecidable. 
The deterministic case is the first known positive decidability result for 
the equivalence problem on two-dimensional
automata over a general alphabet.
We show that there exists a unary three-way deterministic two-dimensional automaton with a
nonregular column projection, and we show that
the row projection of a unary three-way nondeterministic two-dimensional automaton
is always  regular.

\medskip


\medskip

\noindent\textit{MSC2010 classes:} 68Q45 (primary); 20F10 (secondary).
\end{abstract}


\section{Introduction}\label{sec:introduction}

A two-dimensional automaton is a generalization of a one-dimensional 
finite automaton  that operates on two-dimensional input words; that is, on arrays or matrices of symbols from an alphabet $\Sigma$. The two-dimensional automaton model was originally introduced by Blum and Hewitt \cite{BlumHewitt19672DAutomata,BlumHewitt19672DAutomataReport}.

In the one-dimensional case, we may consider either one-way or two-way automaton models. The one-way model is the classical 
definition of a finite automaton, while the two-way model 
allows the input head of the automaton to move both leftward and rightward within the input word. It is well-known that both one-way and two-way automata
recognize the regular languages.

The input head of a two-dimensional automaton can move in four directions,
where the direction is specified by the transition function.
In this paper, we focus on restricted variants of the two-dimensional automaton model. Such restrictions arise from limiting
the movement of the input head of the automaton. If we prevent the input head from moving upward, then we obtain a three-way two-dimensional automaton. If we further prevent the input head from moving leftward,
then we obtain a two-way two-dimensional automaton. The three-way two-dimensional automaton model was introduced by Rosenfeld \cite{Rosenfeld1979PictureLanguages}, and the two-way two-dimensional automaton model was introduced by Dong and Jin \cite{Dong2012TwoWay2DAutomata}.

The emptiness problem for four-way two-dimensional automata
is undecidable \cite{Taniguchi1971DecisionProblemsTwoDimensional},
while the same problem is known to be decidable for three-way deterministic two-dimensional
automata~\cite{InoueTakanami1980DecisionProblems2DAutomata,Petersen19952DTuringMachines}.
Decision problems for two-way two-dimensional automata have not been considered much
in the literature. Since a two-way two-dimensional automaton moves only right and down, it
cannot visit any symbol of the input word more than once. However, the
equivalence problem for two-way two-dimensional automata is, perhaps, not as simple
as one might expect, because the automata tested for equivalence can visit
the input word using a very different strategy and the computations
may partially overlap.

Our results are as follows. 
Using an old result by
Galil~\cite{Galil1976HierarchiesCompleteProblems}, we show
that deciding emptiness of unary three-way nondeterministic two-dimensional automata is
\NP-complete, while emptiness of two-way nondeterministic
two-dimensional automata over general alphabets can be decided in polynomial time.
We show that equivalence of two-way deterministic
two-dimensional automata over general alphabets is decidable, while both equivalence and universality of two-way nondeterministic two-dimensional automata are undecidable. We also consider row
and column projection languages of two-way and three-way two-dimensional automata.

\begin{table}[t]
\centering
\begin{tabular}{c | c c c c c c}
			& \TDFA	& \TNFA	& \TDFATW	& \TNFATW		& \TDFATWOW		& \TNFATWOW \\
\hline
membership	& \cmark	& \cmark	& \cmark		& \cmark			& \cmark			& \cmark \\
emptiness		& \xmark 	& \xmark	& \cmark		& \cmark$^\dagger$	& \cmark$^*$		& \cmark$^*$ \\
universality	& \xmark 	& \xmark	& \cmark		& \xmark			& \cmark			&  \circled{\xmark} \\
equivalence	& \xmark	& \xmark	& \textbf{?}	& \xmark			& \circled{\cmark}	& \circled{\xmark} \\
\end{tabular}
\caption{Decidability results for two-dimensional automaton models.
Decidable problems are marked with \cmark, undecidable problems 
are marked with \xmark, and unknown results 
are marked with a \textbf{?} symbol. 
New decidability results presented in this paper are circled. 
Decision problems for which we provide a complexity bound
are indicated by \cmark$^*$. Decision problems for which we 
provide a complexity bound for the unary case are indicated 
by \cmark$^\dagger$.
}
\label{tab:2Ddecidability}
\end{table}

Table~\ref{tab:2Ddecidability} lists a selection of known 
decidability results for various two-dimensional automaton models. 
Note that almost no problems are decidable for four-way 
two-dimensional automata since neither the emptiness nor universality problems are decidable for that model.
More details about the two-dimensional automaton model and associated problems can be found in survey articles by Inoue and Takanami \cite{Inoue19912DAutomataSurvey} and Kari and Salo \cite{KariSalo2011PictureWalkingAutomataSurvey}, as well as in a recent survey by the first author \cite{Smith2019TwoDimensionalAutomata}.


\section{Preliminaries}

A two-dimensional word consists of a finite array, or rectangle, of cells
labelled by a symbol from a finite alphabet.
The cells around the two-dimensional word are labelled by a special boundary
marker $\#$.
We denote the number of rows (resp., columns) of a two-dimensional word 
$W$ by $|W|_{\text{R}}$ (resp., $|W|_{\text{C}}$).

We begin by defining the deterministic two-dimensional 
automaton model, also known as a four-way deterministic 
two-dimensional automaton.
A two-dimensional automaton has a finite state control and
is capable of moving 
its input head in four directions within a two-dimensional input word: up, down, left, 
and right (denoted $U$, $D$, $L$, and $R$, respectively). 
The squares around the input word are labelled by
the boundary symbol $\#$ and, by 
 remembering the direction of the last move, the input head can be prevented
 from moving outside this boundary. We assume that the machine
 accepts by entering a designated  accept state $q_{\rm accept}$, and the
 machine halts and accepts when it enters $q_{\rm accept}$.
Other equivalent definitions are possible and the precise mode of acceptance
is not important unless one considers questions like state complexity.
Without loss of generality, we can assume that the input head begins its computation in the upper-left corner of the input word.

\begin{definition}[Deterministic two-dimensional automaton]\label{def:2DFA}
A deterministic two-dimensional finite automaton 
(\TDFA) is a tuple $(Q, \Sigma, \delta, q_{0}, q_{\rm accept})$, 
where $Q$ is a finite set of states, $\Sigma$ is the input alphabet 
(with $\# \not\in \Sigma$ acting as a boundary symbol), 
$\delta: (Q \setminus \{q_{\rm accept}\}) 
\times (\Sigma \cup \{\#\}) \to Q \times \{U, D, L, R\}$ 
is the partial transition function, 
and $q_{0}, q_{\rm accept} \in Q$ are the initial and accepting states,
 respectively.
\end{definition}

We can modify a two-dimensional automaton to be nondeterministic 
(\TNFA) in the usual way by changing the transition 
function to map to the power set $2^{Q \times \{U, D, L, R\}}$.

By restricting the movement of the input head, we obtain the 
aforementioned restricted variants of the two-dimensional
automaton model. Our first restriction comes from preventing the input head 
from moving upward.

\begin{definition}[Three-way two-dimensional automaton]
A three-way two-dimensional automaton (\TDFATW/\TNFATW) is a 
tuple $(Q, \Sigma, \delta, q_{0}, q_{\rm accept})$ as in 
Definition~\ref{def:2DFA}, where the transition function $\delta$ is
restricted to use only the directions $\{D, L, R\}$.
\end{definition}

If we prevent the input head from moving both upward and leftward, 
then we get an even more restricted model called a two-way two-dimensional automaton.

\begin{definition}[Two-way two-dimensional automaton]
A two-way two-dimensional automaton (\TDFATWOW/\TNFATWOW) is a 
tuple $(Q, \Sigma, \delta, q_{0}, q_{\rm accept})$ as in 
Definition~\ref{def:2DFA}, where the transition 
function $\delta$ is restricted to use only the directions $\{D, R\}$.
\end{definition}

Both the two-way and three-way automaton variants can be either deterministic or nondeterministic, depending on their transition function $\delta$. 
The power of the two-way two-dimensional automaton model was 
discussed and compared to related automaton models by Dong and Jin \cite{Dong2012TwoWay2DAutomata}. 
The fact that upward and leftward movements are prohibited means 
that the input head can never return to a row if it moves down
or to a column if it moves right. Thus, the two-way two-dimensional 
automaton is a ``read-once" automaton, in the sense that it cannot
visit any symbol twice.

A two-way deterministic two-dimensional automaton cannot visit all symbols of an input
word that has at least two rows and two columns. The same applies
to a given computation of a two-way nondeterministic two-dimensional automaton; however,
different computations of a nondeterministic automaton have the ability
to visit all squares. In fact, it is known that a two-way nondeterministic
two-dimensional automaton cannot be simulated by a three-way deterministic two-dimensional automaton.

\begin{proposition}[Dong and Jin \cite{Dong2012TwoWay2DAutomata}, Kari and Salo \cite{KariSalo2011PictureWalkingAutomataSurvey}]
The recognition power of the two-way nondeterministic two-dimensional automaton model and the three-way deterministic two-dimensional automaton model are incomparable.
\end{proposition}


\section{Language Emptiness}

The language emptiness problem for three-way two-dimensional automata 
is decidable \cite{InoueTakanami1980DecisionProblems2DAutomata,Petersen19952DTuringMachines}.
Using a result by Galil \cite{Galil1976HierarchiesCompleteProblems},
we show that deciding emptiness of unary three-way two-dimensional automata is  \NP-complete.
Galil \cite{Galil1976HierarchiesCompleteProblems} has shown that
deciding emptiness of two-way one-dimensional automata is
in \NP. Note that decidability of emptiness is not obvious
because the tight bound for converting a unary two-way deterministic
one-dimensional automaton to a one-way nondeterministic one-dimensional automaton is
superpolynomial~\cite{Pighizzini2013TwoWayFiniteAutomata}.

\begin{theorem}\label{thm:unary2NFA3Wemptiness}
The emptiness problem for unary three-way nondeterministic two-dimensional automata is \NP-complete.
\end{theorem}

\begin{proof}
Let $\mathcal{A}$ be a unary three-way nondeterministic two-dimensional 
automaton with $n$ states. We restrict the input head of $\mathcal{A}$ 
to operate only on the first row of the input word by replacing all 
downward moves with ``stay-in-place" moves. Call the resulting two-way 
one-dimensional automaton $\mathcal{A}'$. By doubling the number of
states of $\mathcal{A}'$, we can eliminate ``stay-in-place'' moves.

Now,
$L(\mathcal{A}) \neq \emptyset$ if and only if $L(\mathcal{A}')
\neq \emptyset$.
Emptiness of unary two-way one-dimensional automata can be
decided in \NP\ \cite{Galil1976HierarchiesCompleteProblems}.
Furthermore, a unary two-way one-dimensional automaton is a special case of a
unary three-way two-dimensional automaton, and it is known that emptiness
for the former class is \NP-hard 
\cite{Galil1976HierarchiesCompleteProblems}.
\end{proof}

In the general alphabet case,
the emptiness problem for two-way deterministic one-dimensional automata
is \PSPACE-hard \cite{Hunt1973TimeTapeComplexity}, and it follows
that the same applies to deterministic three-way two-dimensional automata.
Equivalence of deterministic and nondeterministic three-way two-dimensional
automata is decidable \cite{Petersen19952DTuringMachines};
however, the known decision algorithm does not operate
in polynomial space.
The question of whether emptiness of deterministic or
nondeterministic three-way two-dimensional automata over general alphabets is 
in \PSPACE\ remains open.


\subsection{Two-Way Two-Dimensional Automata}

The emptiness problem for two-way nondeterministic two-dimensional automata is known 
to be decidable, and the proof of decidability also acts as a trivial proof that the problem is in \NP: simply have the automaton guess an 
accepting computation. It turns out that the problem can be solved in
deterministic polynomial time.

\begin{theorem}
The emptiness problem for two-way nondeterministic 
two-dimensional automata is in \P.
\end{theorem}

\begin{proof}
We can check language emptiness of a two-way nondeterministic two-dimensional
automaton $\mathcal{A}$ via the following procedure:
\begin{enumerate}
\item Beginning in the initial state of $\mathcal{A}$, $q_{0}$, compute the set of states reachable from $q_{0}$. Denote this set by $Q_{\text{reachable}}$.
\item If $q_{\rm accept}$ appears in 
	$Q_{\text{reachable}}$, halt. Otherwise, continue.
\item For each $q \in Q_{\text{reachable}}$, repeat as long as new states get added to
	$Q_{\rm reachable}$:
	\begin{enumerate}
	\item Compute the set of states reachable from $q$. Denote this set by $Q'_{\text{reachable}}$.
	\item If $q_{\rm accept}$ appears in $Q'_{\text{reachable}}$, 
		halt. Otherwise, continue.
	\item Add all states in $Q'_{\text{reachable}}$ to $Q_{\text{reachable}}$ if they do not already occur in that set.
	\end{enumerate}
\item Halt.
\end{enumerate}
If the procedure reaches step 4, then $q_{\rm accept}$ 
was not encountered up to that point and, therefore, the language of
$\mathcal{A}$ is empty. Otherwise, the procedure encountered $q_{\rm accept}$, 
so there exists a sequence of alphabet symbols on 
which the input head of $\mathcal{A}$ can transition from $q_{0}$ to 
$q_{\rm accept}$.

At each stage of the procedure, the set of reachable states is computed by considering all possible transitions on all alphabet symbols from 
the current state. Since $\mathcal{A}$ is a two-way two-dimensional automaton, 
the input head of $\mathcal{A}$ cannot visit the same cell of the input 
word more than once,  which means that at each step both downward and rightward
moves on each alphabet symbol are possible.
If $\mathcal{A}$ has $n$ states, then step 3 is repeated at most $n$ times,
which means that the algorithm terminates in polynomial time.
\end{proof}


\section{Language Equivalence}

Language equivalence is known to be undecidable for four-way deterministic 
two-dimensional automata~\cite{BlumHewitt19672DAutomata}, 
as well as for three-way nondeterministic 
two-dimensional automata \cite{InoueTakanami1980DecisionProblems2DAutomata}.
The equivalence problem for two-way deterministic two-dimensional automata can be expected to
be decidable, but turns out to be perhaps 
not as straightforward as one might
initially assume.


\subsection{Two-Way Deterministic Two-Dimensional Automata}

To obtain the first main result of this section, we use a technical
lemma (Lemma~\ref{lem:equivalencebruteforce}) roughly based
on the following idea.
Suppose that we have a pair of two-way deterministic two-dimensional
automata $\mathcal{A}$ and $\mathcal{B}$, where $\mathcal{A}$ 
has an accepting computation $C_{\mathcal{A}}$ and $\mathcal{B}$ 
has a rejecting computation $C_{\mathcal{B}}$ on some sufficiently large 
input word $W$. Intuitively speaking, our lemma uses a pumping property 
to reduce the dimension of $W$ by finding repeated states in 
$C_{\mathcal{A}}$ and $C_{\mathcal{B}}$.
To do this, we have to be careful
to avoid cases where reducing the size of the input word would force 
the computations to overlap (in parts where they did not originally
overlap) because in such a situation there would be,
in general, no guarantee that the underlying symbols in the overlapping
parts match.

\begin{lemma}\label{lem:equivalencebruteforce}
Let $\mathcal{A}$ and $\mathcal{B}$ be two-way deterministic 
two-dimensional automata with $m$ and $n$ states, respectively. 
Denote $z = m \cdot n \cdot |\Sigma|^2 + 1$ 
and $f(z) = z^2 \cdot (z^2 + z - 1)$.
If $L(\mathcal{A}) - L(\mathcal{B}) \neq \emptyset$, then $L(\mathcal{A}) - L(\mathcal{B})$ contains a two-dimensional word with at most $f(z)$ rows and $f(z)$ columns.
\end{lemma}

\begin{proof}
Consider a two-dimensional word $W \in L(\mathcal{A}) - L(\mathcal{B})$ and suppose that $|W|_{\text{C}} > f(z)$. Let $C_{\mathcal{A}}$ (resp., $C_{\mathcal{B}}$) denote an accepting computation of $\mathcal{A}$ (resp., a rejecting computation of $\mathcal{B}$) on $W$. Without loss of generality, we can assume
that $C_{\mathcal{A}}$ accepts 
(resp., $C_{\mathcal{B}}$ rejects) when
entering a cell containing the border marker; that is, each computation reads
through all columns or all rows of the input word. If the original automata
are allowed to accept/reject inside the input word, then they can be easily
modified to equivalent automata that accept/reject only at border markers.
(Note that a two-way two-dimensional automaton cannot enter an infinite loop.)

We show that $L(\mathcal{A}) - L(\mathcal{B})$ either (i) contains a word with strictly fewer than $|W|_{\text{C}}$ columns and no more than $|W|_{\text{R}}$ rows, or (ii) contains a word with no more than $|W|_{\text{C}}$ columns and strictly fewer than $|W|_{\text{R}}$ rows.

If $C_{\mathcal{A}}$ and $C_{\mathcal{B}}$ share at least $z$ positions in $W$, then two of these shared positions must have been reached via the same states of $\mathcal{A}$ and $\mathcal{B}$ on the same alphabet symbol. These shared positions can be identified and removed to produce a word in $L(\mathcal{A}) - L(\mathcal{B})$ of strictly smaller dimension. (See Figure~\ref{fig:illustration1}.)

Otherwise, $C_{\mathcal{A}}$ and $C_{\mathcal{B}}$ share fewer than $z$ 
positions in $W$. Then, there exists some subword $Z$ in $W$ consisting 
of 
$z \cdot (z^{2} + z - 1)$ consecutive complete 
columns of $W$ where $C_{\mathcal{A}}$ 
and $C_{\mathcal{B}}$ do not intersect.

If $C_{\mathcal{A}}$ and $C_{\mathcal{B}}$ do not enter the 
subword $Z$, then we can reduce the dimension of $W$ without 
affecting either of the computations $C_{\mathcal{A}}$ and 
$C_{\mathcal{B}}$, and a similar reduction is easy to do if
only one of $C_{\mathcal{A}}$ and $C_{\mathcal{B}}$ enters the
subword $Z$.

If $C_{\mathcal{A}}$ and $C_{\mathcal{B}}$ enter the subword $Z$, then without loss of generality, assume that within the subword $Z$, $C_{\mathcal{A}}$ is above $C_{\mathcal{B}}$ and $C_{\mathcal{A}}$ continues to the last row of $Z$. It is possible for $C_{\mathcal{B}}$ to finish earlier if it rejects at the bottom border of $W$. If neither $C_{\mathcal{A}}$ nor $C_{\mathcal{B}}$ 
visit all columns of $Z$, then we can directly
reduce the number of columns of $W$.

If $C_{\mathcal{A}}$ contains a vertical drop of at least $z$ steps within $Z$, or if $C_{\mathcal{A}}$ finishes at least $z$ positions higher than $C_{\mathcal{B}}$ within $Z$---which can occur only when $Z$ consists of the last columns of $W$---then the number of rows of $W$ can be strictly reduced without affecting either of the computations $C_{\mathcal{A}}$ and $C_{\mathcal{B}}$. If such a scenario occurs in the $j$th column of $W$, then $C_{\mathcal{A}}$ contains 
two cells $(i_{1}, j)$ and $(i_{2}, j)$ where $i_{1} < i_{2}$ 
and where the cells
are reached by the same states on the same alphabet symbol, 
and both states and symbol are matched by 
$C_{\mathcal{B}}$ on rows $i_{1}$ and $i_{2}$. Thus, we can reduce the number of rows of $W$ by $i_{2} - i_{1}$. (See Figure~\ref{fig:illustration2}.)
This involves moving the remainder of the computation 
$C_{\mathcal{B}}$ to the left
and adding new cells in the input word to guarantee that
the input word is a rectangle. Note that moving the remainder of $C_{\mathcal{B}}$
to the left cannot force it to overlap with $C_{\mathcal{A}}$.

Now, we know that $C_{\mathcal{A}}$ is above $C_{\mathcal{B}}$ 
within $Z$ and the vertical distance between the two 
computations by the end is at most $z$. Denote by $\max_{Z}$ the maximal vertical difference of $C_{\mathcal{A}}$ and $C_{\mathcal{B}}$ at any fixed column in $Z$. We consider two cases:
\begin{enumerate}
\item Suppose $\max_{Z} \geq z^{2} + z$. Suppose that
	the leftmost value of maximal vertical difference occurs 
	at column $k$ within $Z$. Since, at the end,
	the vertical difference of $C_{\mathcal{A}}$ and $C_{\mathcal{B}}$ is at
	most $z$, and since $C_{\mathcal{A}}$ cannot contain 
	a vertical drop of more than $z$ steps, then there must exist $z$ ``designated" columns between $k$ and the last column of $Z$ where the vertical difference between $C_{\mathcal{A}}$ and $C_{\mathcal{B}}$ either monotonically decreases or stays the same. (See Figure~\ref{fig:illustration3}.)

At two of these ``designated" columns, say $k_{1}$ and $k_{2}$, the states of $C_{\mathcal{A}}$ and $C_{\mathcal{B}}$ and the corresponding alphabet symbols coincide, and we can continue the computation $C_{\mathcal{A}}$ (resp., $C_{\mathcal{B}}$) from column $k_{1}$ in the same way as from column $k_{2}$.

Note that this transformation is not possible
when the vertical distance in the ``designated" columns
is not monotonically decreasing, since if we move 
$C_{\mathcal{A}}$ and $C_{\mathcal{B}}$ so that the computations 
continue from column $k_{1}$ in the same way as from column 
$k_{2}$, the computations could be forced to overlap and we can 
no longer guarantee a matching of alphabet symbols. (See Figure~\ref{fig:illustration4}.)

\item Suppose $\max_{Z} \leq z^{2} + z - 1$. Then $Z$ consists of 
	$z \cdot (z^{2} + z - 1)$ columns, so there must exist $z$ columns where the vertical difference between $C_{\mathcal{A}}$ and $C_{\mathcal{B}}$ is the same. The choice of $z$ implies that, in two of these columns, the states of $C_{\mathcal{A}}$ and $C_{\mathcal{B}}$ and the corresponding alphabet symbols coincide. Thus, we can strictly reduce the number of columns of $W$ without affecting either of the computations $C_{\mathcal{A}}$ and $C_{\mathcal{B}}$.
\end{enumerate}

Altogether, the previous cases establish a method for reducing 
the number of columns of $W$ when $|W|_{\text{C}} > f(z)$. The 
method for reducing the number of rows of $W$ when $|W|_{\text{R}} > f(z)$ 
is completely analogous.
\end{proof}

\begin{figure}[t]
\begin{subfigure}[t]{0.5\textwidth}
\centering
\begin{tikzpicture}[scale=0.55]
	\draw[draw=black] (0,0) rectangle (9,6);
	\draw[fill=black] (0.5,5.5) circle[radius=2pt];
	
	\node at (1.5,4.5) {$\times$};
	
	\node at (2.5,4) {$\times$};
	\draw[->] (1.5,3) -- (2,3.5);
	\node at (1,2.5) {$q_{\mathcal{A}}, q_{\mathcal{B}}$};
	\node at (3,4.5) {\texttt{b}};
	
	\node at (3.5,3.5) {$\times$};
	\node at (4.5,3) {$\times$};
	\node at (5.5,2.5) {$\times$};
	
	\node at (6.5,2) {$\times$};
	\draw[->] (5.5,1) -- (6,1.5);
	\node at (5,0.5) {$q_{\mathcal{A}}, q_{\mathcal{B}}$};
	\node at (7,2.5) {\texttt{b}};
	
	\node at (7.5,1.5) {$\times$};
	
	\draw[->, line width=1pt] (0.5,5.5) -- (1.5,5.5) -- (1.5,4) -- (3.5,4) -- (3.5,3) -- (5.5,3) -- (5.5,2) -- (7.5,2) -- (7.5,0);
	\node at (7.5,-0.75) {$C_{\mathcal{A}}$};
	\draw[->, line width=1pt] (0.5,5.5) -- (0.5,4.5) -- (2.5,4.5) -- (2.5,3.5) -- (4.5,3.5) -- (4.5, 2.5) -- (6.5,2.5) -- (6.5,1.5) -- (9,1.5);
	\node at (9.75,1.5) {$C_{\mathcal{B}}$};
\end{tikzpicture}
\caption{Removing shared computations}
\label{fig:illustration1}
\end{subfigure}
\begin{subfigure}[t]{0.5\textwidth}
\centering
\begin{tikzpicture}[scale=0.55]
	\draw[draw=black] (0,0) rectangle (9,6);
	\draw[fill=black] (0.5,5.5) circle[radius=2pt];
	
	\draw[fill=black] (1,3.75) circle[radius=2pt];
	\draw[fill=black] (8,3.75) circle[radius=2pt];
	\draw[dashed] (1,3.75) -- (8,3.75);
	
	\draw[fill=black] (4,2) circle[radius=2pt];
	\draw[fill=black] (8,2) circle[radius=2pt];
	\draw[dashed] (4,2) -- (8,2);
	
	\node at (6,2.9) {$q_{\mathcal{A}}$, \texttt{b}};
	\draw[->] (7,2.9) -- (7.75,3.5);
	\draw[->] (7,2.9) -- (7.75,2.25);

	\node at (1,1.5) {$q_{\mathcal{B}}$, \texttt{b'}};
	\draw[->] (1,2) -- (1,3.25);
	\draw[->] (1.75,1.5) -- (3.75,2);
	
	\node at (9.75,3.75) {$(i_{1}, j)$};
	\draw[->] (8.75,3.75) -- (8.25,3.75);
	\node at (9.75,2) {$(i_{2}, j)$};
	\draw[->] (8.75,2) -- (8.25,2);
	
	\draw[->, line width=1pt] (0.5,5.5) -- (1,5.5) -- (1,5) -- (4,5) -- (4,4.5) -- (8,4.5) -- (8,1);
	\node at (9.75,1) {$C_{\mathcal{A}}$};
	\draw[->, line width=1pt] (0.5,5.5) -- (0.5,4) -- (1,4) -- (1,3.5) -- (2,3.5) -- (2,2.5) -- (4,2.5) -- (4,1.5) -- (7,1.5) -- (7,0);
	\node at (7,-0.75) {$C_{\mathcal{B}}$};
\end{tikzpicture}
\caption{Removing rows within vertical drop}
\label{fig:illustration2}
\end{subfigure}
\begin{subfigure}[t]{0.5\textwidth}
\centering
\begin{tikzpicture}[scale=0.55]
	\draw[draw=black] (0,0) rectangle (9,6);
	\draw[fill=black] (0.5,5.5) circle[radius=2pt];
	
	\draw[dashed] (1.5,0) -- (1.5,6);
	\node at (1.5,-0.5) {$k$};
	
	\draw[<->, dashed] (2,5.5) -- (2,1.5);
	\node at (3.6,3) {\shortstack{{\footnotesize$\max_{Z} \geq$} \\ {\footnotesize$z^{2} + z$}}};
	
	\draw[<->, dashed] (7.75,1.5) -- (7.75,0.5);
	\node at (8.3,1) {{\footnotesize$\leq z$}};
	
	\draw[->, line width=1pt] (0.5,5.5) -- (2.5,5.5) -- (2.5,4.5) -- (4,4.5) -- (4,4) -- (6,4) -- (6,2.5) -- (7.5,2.5) -- (7.5,1.5) -- (9,1.5);
	\node at (9.75,1.5) {$C_{\mathcal{A}}$};
	\draw[->, line width=1pt] (0.5,5.5) -- (0.5,4) -- (1.5,4) -- (1.5,1.5) -- (3,1.5) -- (3,1) -- (7,1) -- (7,0.5) -- (9,0.5);
	\node at (9.75,0.5) {$C_{\mathcal{B}}$};
\end{tikzpicture}
\caption{Removing columns after row removal}
\label{fig:illustration3}
\end{subfigure}
\begin{subfigure}[t]{0.5\textwidth}
\centering
\begin{tikzpicture}[scale=0.55]
	\draw[draw=black] (0,0) rectangle (9,6);
	\draw[fill=black] (0.5,5.5) circle[radius=2pt];
	
	\draw[fill=black] (1.25,5.5) circle[radius=2pt];
	\draw[fill=black] (1.25,5) circle[radius=2pt];
	\draw[dashed] (1.25,0) -- (1.25,6);
	\node at (1.25,-0.5) {$k_{1}$};
	
	\draw[fill=black] (5.5,3.5) circle[radius=2pt];
	\draw[fill=black] (5.5,0.5) circle[radius=2pt];
	\draw[dashed] (5.5,0) -- (5.5,6);
	\node at (5.5,-0.5) {$k_{2}$};
	
	\draw[->, line width=1pt] (0.5,5.5) -- (2.5,5.5) -- (2.5,4) -- (4.5,4) -- (4.5,3.5) -- (6.5,3.5) -- (6.5,2) -- (9,2);
	\node at (9.75,2) {$C_{\mathcal{A}}$};
	\draw[->, line width=1pt] (0.5,5.5) -- (0.5,5) -- (2,5) -- (2,0.5) -- (9,0.5);
	\node at (9.75,0.5) {$C_{\mathcal{B}}$};
\end{tikzpicture}
\caption{Situation with non-removable columns}
\label{fig:illustration4}
\end{subfigure}
\caption{Illustrations depicting various scenarios in Lemma~\ref{lem:equivalencebruteforce}}
\label{fig:lemmaillustrations}
\end{figure}
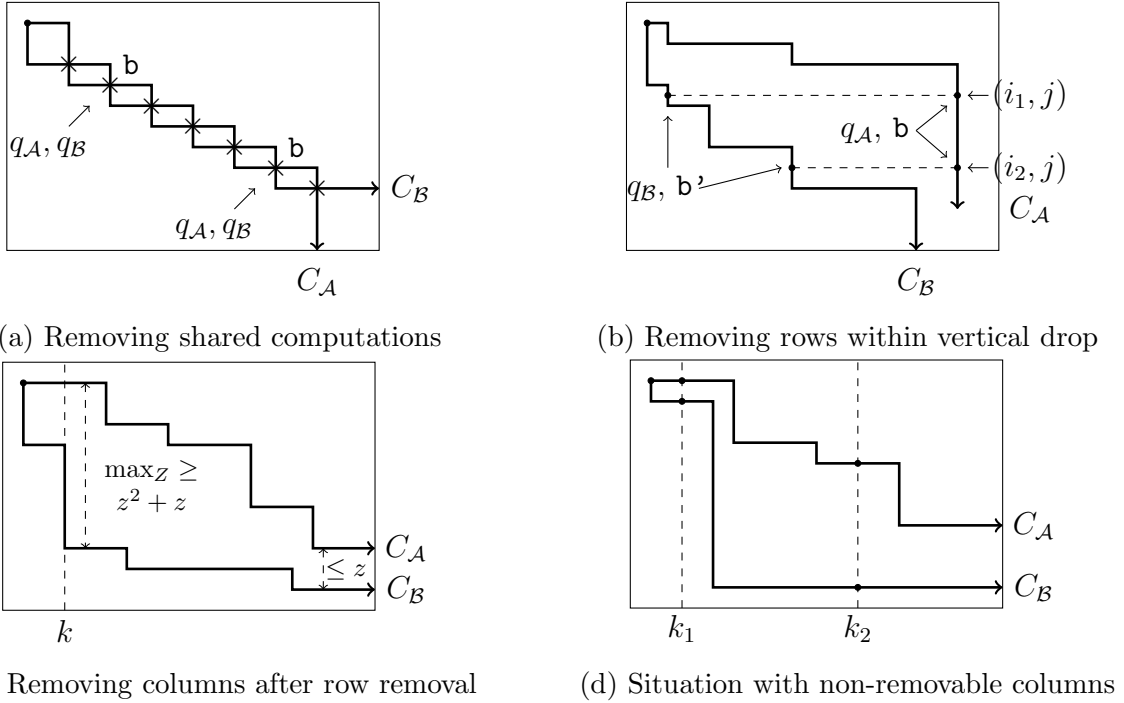

Lemma~\ref{lem:equivalencebruteforce} gives a brute-force 
algorithm to decide equivalence for two-way deterministic 
two-dimensional automata by checking all input words up to a 
given dimension, and the algorithm depends only on the two automata in question. 
As a consequence of the existence of such an algorithm, we obtain the following result.

\begin{theorem}
The equivalence problem for two-way deterministic two-dimensional automata
over a general alphabet is decidable.
\end{theorem}

Lemma~\ref{lem:equivalencebruteforce} also decides the inclusion problem and it follows that inclusion of two-way deterministic two-dimensional automata is decidable.

However, the brute-force
algorithm given by Lemma~\ref{lem:equivalencebruteforce} is extremely
inefficient.
The question of whether there exists a more efficient decidability procedure remains open.

Note that a two-way deterministic two-dimensional automaton cannot visit a 
symbol in the input word more than once, so we can reduce the number of states of such an automaton in a manner analogous to reducing states in a one-way deterministic one-dimensional automaton: mark pairs of states $(q_{i}, q_{j})$ as distinguishable if one of $q_{i}$ and $q_{j}$ is final and the other is not, then iteratively mark pairs of states as distinguishable if both states reach a previously-marked pair on some transition on the same alphabet symbol and in the same direction (downward or rightward). Such an
algorithm runs in polynomial time; however, it is easy to
see that a state-minimal 
two-way deterministic two-dimensional automaton need not be 
unique. Therefore, it is unclear whether a state minimization 
approach can be used to decide the equivalence problem.


\subsection{Two-Way Nondeterministic Two-Dimensional Automata}

For the second main result of this section, we show that the equivalence problem for two-way nondeterministic two-dimensional automata is undecidable. Before we give a proof of this claim, we require a few preliminary definitions and results.

A configuration $C$ of a machine $\mathcal{M}$ is a sequence of tape symbols of $\mathcal{M}$, where the currently-scanned symbol is represented by the subscripted current state $q$. A computation history is a sequence of configurations $C_{1}, C_{2}, \dots, C_{m}$ that a machine goes through as it performs a computation. A computation table is a computation history where each configuration is written one on top of another, the first row of the table is the initial configuration, and the last row of the table is the final configuration.

To show that the equivalence problem for two-way nondeterministic two-dimensional automata is undecidable, we will reduce the problem to showing that the universality problem for two-way nondeterministic two-dimensional automata is undecidable. We do so by constructing a two-way nondeterministic two-dimensional automaton $\mathcal{A}$ that checks whether or not a given computation table corresponds to an accepting computation history for some deterministic linear-bounded automaton $\mathcal{M}$.

In order to perform such a reduction, we must represent the configurations in the computation history of $\mathcal{M}$ using a ``double encoding" technique. A double encoding of a configuration denotes symbols and states as pairs., where the last symbol or state of the $i$th pair is the same as the first symbol or state of the $(i+1)$st pair.

\begin{example}
Consider the configuration $C = \texttt{a} \ \texttt{b} \ \texttt{c}_{q} \ \texttt{d} \ \texttt{e} \ \texttt{f}$, where the current state is $q$ and the currently-scanned symbol is $c$. The double encoding of this configuration is
\begin{equation*}
(\texttt{a}, \texttt{b}), (\texttt{b}, \texttt{c}_{q}), (\texttt{c}_{q}, \texttt{d}), (\texttt{d}, \texttt{e}), (\texttt{e}, \texttt{f}).
\end{equation*}
\end{example}

With this notion, we can now present the proof of the result.

\begin{theorem}
The universality problem for two-way nondeterministic two-dimensional automata is undecidable.

\begin{proof}
Construct a two-way nondeterministic two-dimensional automaton $\mathcal{A}$ that takes as input a computation table. The automaton $\mathcal{A}$ checks that its input does not correspond to a valid computation history for some deterministic linear-bounded automaton $\mathcal{M}$ in the following way:
\begin{enumerate}
\item Nondeterministically check one of the following properties:
	\begin{enumerate}
	\item The $i$th row of the input is not a double encoding of a configuration of $\mathcal{M}$. This can be done by moving nondeterministically to the $i$th row of the input and reading the row as a nondeterministic one-dimensional automaton.
	\item The first row of the input is not a double encoding of an initial configuration of $\mathcal{M}$.
	\item The last row of the input is not a double encoding of a final configuration of $\mathcal{M}$.
	\item For some row $i$, the $(i+1)$st row corresponding to the configuration $C_{i+1}$ does not follow from the $i$th row corresponding to the configuration $C_{i}$ in one computation step of $\mathcal{M}$.
	
	For the last property, we have two possibilities:
	\begin{enumerate}
	\item A symbol in $C_{i+1}$ was changed, and that symbol was not read by a state in $C_{i}$. This possibility is checked by moving nondeterministically to the column of that symbol in $C_{i}$ and making a downward move.
	\item A symbol in $C_{i+1}$ was changed by a state, but the change does not correspond to a valid computation step of $\mathcal{M}$. Since $\mathcal{A}$ cannot make a leftward move and hence cannot replicate leftward moves made by $\mathcal{M}$, we require the double encoding to check for such an invalid computation step.
	\end{enumerate}
	\end{enumerate}
\item If any of the properties are satisfied, accept.
\end{enumerate}
Using this procedure, $\mathcal{A}$ accepts all inputs if and only if $L(\mathcal{M}) = \emptyset$; that is, if and only if $\mathcal{M}$ has no accepting computations. Since the emptiness problem for deterministic linear-bounded automata is undecidable \cite{Sipser1997Computation}, the result follows.
\end{proof}
\end{theorem}

Lastly, since the universality problem is the same as checking whether $L(\mathcal{A}) = \Sigma^{*}$ for some automaton $\mathcal{A}$, we get the following corollary.

\begin{corollary}
The equivalence problem for two-way nondeterministic two-dimensional automata is undecidable.
\end{corollary}

Another hard open problem
is to determine whether or not equivalence of three-way deterministic two-dimensional automata
is decidable.


\section{Row and Column Projection Languages}

The row projection (resp., column projection)
of a two-dimensional language $L$ is the one-dimensional language consisting
of the first rows (resp., first columns) of all two-dimensional words in
$L$.

General (four-way) deterministic two-dimensional automata
can recognize that the input word has, for example,
exponential or doubly-exponential side-length
\cite{KariMoore2004RectanglesAndSquares},
which implies that the row or column projections,
even in the unary case, need not be context-free.

Kinber \cite{Kinber1985ThreeWayAutomataOneLetter} has shown that 
the numbers of rows and columns of  unary two-dimensional words in a language recognized
by a three-way deterministic two-dimensional automaton are connected by certain
bilinear forms. 

Here, we consider the row and column projection languages of unary languages 
recognized by three-way  two-dimensional automata, and
 we get differing regularity 
results for the row and column projection languages, respectively.

\begin{theorem}\label{thm:unary2NFA3Wrowprojection}
Given a unary three-way nondeterministic two-dimensional 
automaton $\mathcal{A}$, the row projection language of 
$L(\mathcal{A})$ is regular.
\end{theorem}

\begin{proof}
Let $\mathcal{A}$ be a unary three-way nondeterministic 
two-dimensional automaton. The row projection language of 
$\mathcal{A}$ can be recognized by a two-way nondeterministic 
one-dimensional automaton $\mathcal{B}$ by simulating the 
computation of $\mathcal{A}$ and replacing downward moves by 
``stay-in-place" moves. The automaton $\mathcal{B}$ accepts
if the simulated computation of $\mathcal{A}$ accepts, and
hence $\mathcal{B}$ recognizes exactly the first rows
of two-dimensional words of $L(\mathcal{A})$.

Since the row projection language of $L(\mathcal{A})$ is 
recognized by a two-way nondeterministic one-dimensional automaton, 
it is regular.
\end{proof}

For the column projection operation, however, we are not 
guaranteed to have regularity even with unary three-way deterministic 
two-dimensional automata. As a counterexample, we use the
following  unary language:
\begin{equation*}
L_{\text{composite}} = \{a^{m} \ \mid \ m > 1 \text{ and } m \text{ is not prime}\}.
\end{equation*}
Clearly, $L_{\text{composite}}$ is nonregular since its complement
is nonregular.
The nonregularity of $L_{\text{composite}}$ plays a 
key role in the following lemma.

\begin{figure}
\centering
\begin{tikzpicture}
\matrix[matrix of nodes,nodes={inner sep=0pt,text width=.75cm,align=center,minimum height=.65cm}]{
\node(a00){\#};	& \#					& \#					& \#					& \# \\
\#			& \node(a11){$a_{11}$};	& \node(a21){$a_{12}$};	& \node(a31){$a_{13}$};	& \# \\
\#			& \node(a12){$a_{21}$};	& \node(a22){$a_{22}$};	& \node(a32){$a_{23}$};	& \# \\
\#			& \node(a13){$a_{31}$};	& \node(a23){$a_{32}$};	& \node(a33){$a_{33}$};	& \# \\
\#			& \node(a14){$a_{41}$};	& \node(a24){$a_{42}$};	& \node(a34){$a_{43}$};	& \node(a44){$\#$}; \\
\#			& \node(a15){$a_{51}$};	& \node(a25){$a_{52}$};	& \node(a35){$a_{53}$};	& \# \\
			& \node(a16){\phantom{$a_{16}$}};	& \vdots				& 					& \\
};

\draw[->] (a11.center)+(2mm,-2mm) -- (a22.north west);
\draw[->] (a22.center)+(2mm,-2mm) -- (a33.north west);
\draw[->] (a33.center)+(2mm,-2mm) -- ([yshift=-0.5mm] a44.north west);
\draw[->] (a44.west)+(2mm,0mm) -- ([xshift=-0.5mm, yshift=0.45mm] a34.east);
\draw[->] (a34.south west)+(2mm,2mm) -- (a25.north east);
\draw[->] (a25.south west)+(2mm,2mm) -- ([yshift=1.7mm] a16.north east);
\end{tikzpicture}
\caption{An illustration of the movement of the input head of the automaton $\mathcal{C}$, constructed in Lemma~\ref{lem:unary2NFA3Wcolumncomposite}, on an input word with three columns}
\label{fig:unary2NFA3Wcolumncomposite}
\end{figure}
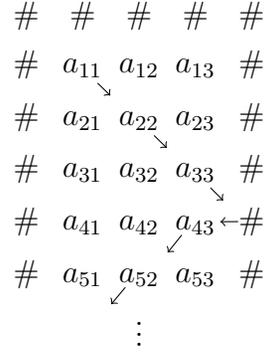

\begin{lemma}\label{lem:unary2NFA3Wcolumncomposite}
There exists a unary three-way deterministic 
two-dimensional automaton $\mathcal{C}$ such that the column projection language of $L(\mathcal{C})$ is equal to $L_{\text{composite}}$. 
\end{lemma}

\begin{proof}
We construct the automaton $\mathcal{C}$ as follows. 
Given an input word of dimension $m \times n$, 
 $\mathcal{C}$  first verifies that $m > 1$; that is, 
 that the input word has more than one row. Then, moving in a diagonal manner from the upper-left corner of the input word, the input head of $\mathcal{C}$ travels downward and rightward until it reaches the last symbol of the $(m+1)$st row. From this symbol, the input head of $\mathcal{C}$ travels downward and leftward until it reaches the first symbol of the $(2m+1)$st row. These movements are illustrated in Figure~\ref{fig:unary2NFA3Wcolumncomposite}, where $\searrow$ denotes a combined downward/rightward move and $\swarrow$ denotes a combined downward/leftward move.

The automaton $\mathcal{C}$ accepts the input word if, after making at least two sweeps across the input word, the input head reaches the lower-left or lower-right corner of the input word after completing its current sweep.

The input head of $\mathcal{C}$ is able to detect when it has reached the lower-left or lower-right corner of the input word in the following way:
\begin{itemize}
\item If the input head reads \# following a leftward move, make a downward move followed by a rightward move and check that both symbols read are \#. If so, accept. Otherwise, if the second symbol read is not \#, continue.
\item If the input head reads \# following a rightward move, make a downward move followed by a leftward move and check that both symbols read are \#. If so, accept. Otherwise, if the second symbol read is not \#, continue.
\item If the input head reads \# following a downward move, reject.
\end{itemize}

Following this construction, we see that the column projection 
language of $L(\mathcal{C})$ consists of all strings of length at 
least $2$ that do not have prime length; that is, $L_{\text{composite}}$.
The computation of $\mathcal{C}$ is completely deterministic.
\end{proof}

Using Lemma~\ref{lem:unary2NFA3Wcolumncomposite}, we obtain the main result pertaining to column projection languages.

\begin{theorem}
Given a unary three-way deterministic two-dimensional automaton $\mathcal{A}$, the column projection language of $L(\mathcal{A})$ is not always regular.
\end{theorem}

\begin{remark}
In a classical work, Greibach used the language $L_{\text{composite}}$ to show that one-way nondeterministic checking stack automata can recognize nonregular unary languages \cite{Greibach1969CheckingAutomata}.
\end{remark}

Ibarra et al.~\cite{Ibarra2018AcceptingRunsTwoWay} introduced the notion of 
an accepting run of a two-way automaton. An accepting run is,
roughly speaking, a sequence of 
states that the automaton enters during the course of some accepting 
computation. They showed that the set of accepting runs of a two-way 
automaton can be strongly nonregular.

The proof of Lemma~\ref{lem:unary2NFA3Wcolumncomposite} provides
an example where the set of accepting runs of a unary two-way
nondeterministic automaton is not  regular.
Using the automaton $\mathcal{C}$ from 
Lemma~\ref{lem:unary2NFA3Wcolumncomposite}, simulate the computation of
$\mathcal{C}$ with a two-way one-dimensional automaton $\mathcal{B}$ as in
Theorem~\ref{thm:unary2NFA3Wrowprojection}. Then, the set of
accepting runs of $\mathcal{B}$ will not be regular because the number
of ``stay-in-place'' moves is guaranteed to be a composite number.
Note that, although $\mathcal{C}$ is deterministic, the simulating
two-way one-dimensional automaton $\mathcal{B}$ will be nondeterministic
because it has to guess when $\mathcal{C}$ has reached the last row
and when the computation should accept.


\subsection{Two-Way Two-Dimensional Automata over General Alphabets}

As opposed to the three-way case, we can establish regularity results for both the row projection and column projection languages of two-way nondeterministic two-dimensional automata. Furthermore, we no longer require that the automaton has a unary alphabet.

\begin{theorem}\label{thm:twowayrowprojectionregular}
Given a two-way nondeterministic two-dimensional automaton $\mathcal{A}$, the row projection language of $\mathcal{A}$ is regular.
\end{theorem}

\begin{proof}
Let $\mathcal{A}$ be a two-way nondeterministic two-dimensional automaton. 
Construct a one-way nondeterministic one-dimensional automaton $\mathcal{B}$
(with ``stay-in-place'' moves)
to recognize the row projection language of $\mathcal{A}$ as follows:
\begin{enumerate}
	\item Use $\mathcal{B}$ to
		nondeterministically
		simulate rightward moves of 
	$\mathcal{A}$.
\item Simulate state changes of $\mathcal{A}$ after a downward move via
	a ``stay-in-place'' move, and
	nondeterministically guess the alphabet character that $\mathcal{A}$
	reads on the next row.
\item After $\mathcal{A}$ makes the first downward move,  begin
	simulating rightward moves of $\mathcal{A}$ by moving right but
	nondeterministically selecting an alphabet character for
	the simulated transition of $\mathcal{A}$.
	After this point, $\mathcal{B}$ ignores its own input.
\end{enumerate}
Note that, after $\mathcal{A}$ makes a downward move, it can never return to the previous row. Therefore, we do not care about the remaining contents of the previous row.

From here, $\mathcal{B}$ must count the number of 
rightward moves it simulates and check that there exists 
an accepting computation of $\mathcal{A}$ where the number 
of rightward moves and the length of the
remaining input to $\mathcal{B}$ are equal.

Since the row projection language of $\mathcal{A}$ is recognized by a nondeterministic one-dimensional automaton, it is regular.
\end{proof}

As opposed to the three-way case, the handling of
rows and columns for two-way two-dimensional automata is symmetric: 
a row or column is read one way and
cannot be returned to after a downward or rightward move, respectively.
Using a completely analogous construction as in the proof of 
Theorem~\ref{thm:twowayrowprojectionregular}, we obtain a similar 
result for column projection languages.

\begin{theorem}
Given a two-way nondeterministic two-dimensional automaton $\mathcal{A}$, the column projection language of $\mathcal{A}$ is regular.
\end{theorem}


\section{Conclusion}

In this paper, we considered decision problems for three-way
and two-way two-dimensional automata. 
We showed that the language emptiness problem is \NP-complete 
for unary three-way nondeterministic two-dimensional automata 
and in \P\ for two-way nondeterministic two-dimensional 
automata over a general alphabet. 
We also proved that the language equivalence problem is decidable 
for two-way deterministic two-dimensional automata, while the equivalence and universality problems in the nondeterministic case are undecidable.  Lastly, we investigated the row projection and column projection operations and found that the resulting languages are regular for two-way nondeterministic two-dimensional automata
over a general alphabet. 
In the three-way case, only the row projection of a unary two-dimensional
language is regular.

As mentioned throughout this paper, some open problems remain in this area of study. For three-way two-dimensional automata, it is unknown whether the general-alphabet emptiness problem belongs to \PSPACE. A positive result would imply that the problem is \PSPACE-complete. For two-way
two-dimensional automata, it could be interesting
to investigate whether an efficient algorithm exists to decide the equivalence problem in the deterministic case 
(possibly by using a state minimization approach). Table~\ref{tab:2Ddecidability}
in Section~\ref{sec:introduction}
lists a selection of decidability questions for various two-dimensional automaton models that remain unresolved, and for most problems listed as decidable in Table~\ref{tab:2Ddecidability}, exact complexity bounds have not yet been determined.


\bibliographystyle{plain}
\bibliography{References}


\end{document}